\documentclass[
twocolumn,
showpacs,preprintnumbers,amsmath,amssymb,aps,pra,
superscriptaddress
]{revtex4-1}

\usepackage{graphicx}
\usepackage{bm}
\usepackage{amsfonts}
\usepackage{amscd}
\usepackage{amsmath}    
\usepackage{enumerate}
\usepackage{subfigure}

\usepackage{amssymb}
\usepackage{amsthm}

\usepackage{hyperref}

\newtheorem{lemma}{Lemma}

\newtheorem{definition}{Definition}

\newtheorem{remark}{Remark}

\newcommand{\bra}[1]{\mbox{$\left\langle #1 \right|$}}
\newcommand{\ket}[1]{\mbox{$\left| #1 \right\rangle$}}
\newcommand{\braket}[2]{\mbox{$\left\langle #1 | #2 \right\rangle$}}

\def\tr{{\rm Tr}}
\def\IR{{\mathbb R}}
\def\IC{{\mathbb C}}
\def\sizeofU{{r}}

\newcommand{\munorm}[1]{\lVert #1 \rVert_{\vec{\mu}}}
\newcommand{\maxnorm}[1]{\lVert #1 \rVert_{\text{max}}}
\newcommand{\sumnorm}[1]{\lVert #1 \rVert_{\text{sum}}}

\def\KrausDim{{d}}
\def\ChannelDim{{n}}
\DeclareMathOperator{\myUgrp}{U}

\def\myrm{}

\def\dmathX#1#2{
$$\lineskiplimit=1000pt \advance\lineskip by #1\jot 
\mathsurround=0pt \tabskip=0pt plus 1000pt
\everycr{\noalign{\penalty\interdisplaylinepenalty}}
\halign to \displaywidth{
\hfil$\displaystyle{##}$\tabskip=0pt&%
\hfil $\displaystyle{{}##{}}$\hfil &%
\hfil $\displaystyle{{}##{}}$\hfil &%
$\displaystyle{##}$\hfil \tabskip=0pt plus 1000pt minus 1000pt&%
\refstepcounter{equation}\label{##}\llap{(\theequation)}\tabskip=0pt\cr
\noalign{\ifdim \prevdepth>-1000pt \vskip -#1\jot\fi}
#2\crcr}$$}

\begin{document}

\title
{Time-Energy Measure for Quantum Processes}

\author{Chi-Hang Fred Fung}
\email{chffung@hku.hk}
\affiliation{Department of Physics and Center of Theoretical and Computational Physics, University of Hong Kong, Pokfulam Road, Hong Kong}
\author{H.~F. Chau}
\affiliation{Department of Physics and Center of Theoretical and Computational Physics, University of Hong Kong, Pokfulam Road, Hong Kong}

\begin{abstract}
Quantum mechanics sets limits on how fast quantum processes can run given some system energy through time-energy uncertainty relations, and they imply that time and energy are tradeoff against each other.
Thus, we propose to measure the time-energy as a single unit for quantum channels.
We consider a time-energy measure for quantum channels and compute lower and upper bounds of it using the channel Kraus operators.
For a special class of channels (which includes the depolarizing channel), we can obtain the exact value of the time-energy measure.
One consequence of our result is that erasing quantum information requires 
$\sqrt{(n+1)/n}$ 
times more time-energy resource than erasing classical information, where $n$ is the system dimension.
\end{abstract}


\maketitle

\section{Introduction}

Evolution of quantum processes (including performing quantum computation)
requires physical resources, in particular, time and energy.
The computation speed of a physical device is governed by physical laws and is limited by the energy of the device.
Under the constraints of quantum mechanics, system evolutions are bounded by
time-energy uncertainty relations (TEURs)~\cite{Lloyd2000}.
The investigation of TEURs has a long history.
The first major result of a TEUR was proved by Mandelstam and Tamm~\cite{Mandelstam1945}.
This was followed by 
subsequent work on isolated systems
~\cite{Bhattacharyya1983,Anandan1990,Uhlmann1992,Vaidman1992,Pfeifer1993,Margolus1996,Margolus1998,Chau2010}
and 
composite systems with entanglement \cite{Giovannetti2003,Giovannetti2003b,Zander2007}.
Recently, TEURs for general quantum processes have also been proved~\cite{Taddei2013,delCampo2013}.
The general form of TEURs is an inequality that sets a lower limit on 
the product of the system energy (or a function of the energies) and the time it takes to evolve an initial state to a final state (e.g., an orthogonal state).
Motivated by the TEURs and recognizing that time and energy are tradeoff against each other,
time-energy can be regarded as a single property
of a quantum process.
The intuition is that the more computation or work a quantum process performs, the more time-energy it requires.
And it is up to the system designer (or nature) to perform it
with more time but less energy, or vice versa.
Thus, our goal in this paper is to 
investigate the time-energy requirements of quantum processes by using a time-energy measure.
Chau~\cite{Chau2011} proposed a time-energy measure for unitary transformations that is based on a TEUR proved earlier~\cite{Chau2010}.
In this paper, we extend this measure 
to quantum processes.
The TEUR due to Chau~\cite{Chau2010} is tight 
in the sense that it can be saturated by some states and Hamiltonians, and thus it serves to motivate a good definition for a time-energy measure.
To see this, let's start with this TEUR.
Given a time-independent Hamiltonian $H$ of a system, 
the time 
$t$
needed to to evolve a state $\ket{\Phi}$ under the action of $H$ to a state 
whose fidelity~\footnote{We adopt the fidelity definition $F(\rho,\sigma)=\big(\tr \sqrt{\rho^{1/2} \sigma \rho^{1/2}}\big)^2$ for two quantum states $\rho$ and $\sigma$.} is less than or equal to $\epsilon$ satisfies the TEUR
\begin{align}
\label{eqn-time-energy-relation1}
t
\geq \frac{(1-\sqrt{\epsilon})\hbar}{A \sum_j |\alpha_j|^2 |E_j|}
\end{align}
where $E_j$'s are the eigenvalues of $H$ with the corresponding normalized energy eigenvectors $\ket{E_j}$'s, $\ket{\Phi}=\sum_j \alpha_j \ket{E_j}$, and $A \approx 0.725$ is a universal constant.
Essentially, 
after time 
$t$,
the state transforms unitarily according to $U=e^{-i Ht/\hbar}$.
The same $U$ could be implemented with either a high energy $H$ run for a shorter time or a low energy $H$ run for a longer time.
Based on Eq.~\eqref{eqn-time-energy-relation1}, a weighted sum of $|t E_j|$'s 
serves as
an indicator of the time-energy resource needed to perform $U$,
and as such the following time-energy measure on unitary matrices was proposed by Chau~\cite{Chau2011}:
\begin{align*}
\munorm{U}&=\sum_{j=1}^r \mu_j |\theta_j|^\downarrow
\end{align*}
where $U$ has eigenvalues $\exp(-i E_j t/\hbar)\equiv \exp(\theta_j)$ and $\vec{\mu}$ is some fixed vector to be described later (see Sec.~\ref{sec-problem-formulation}). 
In essence, a large value of $\munorm{U}$ suggests that a long time may be needed to run a Hamiltonian that implements $U$ for a fixed energy, and vice versa.

In this paper,
we are interested in 
an analogous measure
for quantum channels which include unitary transformations as special cases.
We are given a quantum channel $\mathcal{F}(\rho)$ acting on system $A$ that maps $n \times n$ density matrix $\rho$ to another one with the same dimension.
There exist unitary extensions $U_{BA}$ in a larger Hilbert space with an ancillary system $B$ such that 
$\mathcal{F}(\rho) = \tr_B [ U_{BA} (\ket{0}_B\bra{0} \otimes  \rho_A) U_{BA}^\dag ]$.
Each $U_{BA}$ could have a different time-energy spectrum and 
we want to select the one requiring the least resource 
for 
$\mathcal{F}$.
We extend the resource indicator for $U$ to quantum channel $\mathcal{F}$ by defining 
\dmathX2{
\munorm{\mathcal{F}} &\equiv& \min_U & \munorm{U}  \cr
&&
\text{s.t.} &
\mathcal{F}(\rho) = \tr_B [ U_{BA} (\ket{0}_B\bra{0} \otimes  \rho_A) U_{BA}^\dag ]
\: \forall \rho.
\cr
}
This gives a $U$ that consumes the least time-energy resource.
Thus, $\munorm{\mathcal{F}}$ is an indicator of the resource needed to perform $\mathcal{F}$.
We formally formulate 
this problem in Sec.~\ref{sec-problem-formulation} into the ``partial $U$ problem'' and the ``channel problem''.
Then, we simplify the ``partial $U$ problem'' in Sec.~\ref{sec-simplification} and solve the a special case of it in Sec.~\ref{sec-single-vector-transformation}.
The special case solution will be used to prove
our major results
which are 
the upper bound of 
the time-energy resource measure $\munorm{\mathcal{F}}$ 
(Sec.~\ref{sec-upper-bound}), the lower bound of $\munorm{\mathcal{F}}$ (Sec.~\ref{sec-lower-bound}), and the optimal 
$\munorm{\mathcal{F}}$ for a class of quantum channels (Sec.~\ref{sec-optimal-energy}).
The lower- and upper-bounds of 
the time-energy $\munorm{\mathcal{F}}$ hold
for any quantum channel $\mathcal{F}$ and for specific $\vec{\mu}$:
\begin{align}
\maxnorm{\mathcal{F}} 
&\geq
\min_{\mathbf{v}: \: \lVert \mathbf{v} \rVert \leq 1}
\:
\max_{1\leq i \leq n}
\cos^{-1} 
\Big[
\operatorname{Re}(\lambda_i(\sum_{j=1}^d v_j F_j)) 
\Big]
\\
\maxnorm{\mathcal{F}} 
&\leq 
\min_{\mathbf{v}: \: \lVert \mathbf{v} \rVert \leq 1}
\sum_{i=1}^n 
\cos^{-1}
\Big[
\operatorname{Re}
(\lambda_i(\sum_{j=1}^d v_j F_j))
\Big]
\\
\sumnorm{\mathcal{F}} &\geq 
\min_{\mathbf{v}: \: \lVert \mathbf{v} \rVert=1}
\:
\max_{1\leq i \leq n}
2 \cos^{-1} |\lambda_i(\sum_{j=1}^d v_j F_j)| 
\\
\sumnorm{\mathcal{F}} &\leq 
\min_{\mathbf{v}: \: \lVert \mathbf{v} \rVert \leq 1}
\sum_{i=1}^n 
2\cos^{-1}
\Big[
\operatorname{Re}
(\lambda_i(\sum_{j=1}^d v_j F_j))
\Big]
\end{align}
where $\lVert \mathbf{v} \rVert=\sqrt{\sum_{j=1}^d |v_j|^2}$,
$F_j \in \IC^{n\times n}, j=1,\dots,d$ are the Kraus operators of $\mathcal{F}$,
and
$\lambda_i(\cdot)$ denotes the $i$th eigenvalue of its argument.
Here, $\maxnorm{\cdot}$ is a short-hand notation for $\munorm{\cdot}$ with $\vec{\mu}=[1,0,0,\dots]$ and $\sumnorm{\cdot}$ for $\munorm{\cdot}$ with $\vec{\mu}=[1,1,1,\dots]$.
For a class of channels (which includes the depolarizing channel), we obtain the exact value for $\maxnorm{\mathcal{F}}$ in Sec.~\ref{sec-optimal-energy}. 
In particular, when $\mathcal{F}$ is a depolarizing channel with probability $q$ that the input state is unchanged, its time-energy requirement is $\maxnorm{\mathcal{F}}=\cos^{-1}\sqrt{ q + (1-q)/n^2 }$. 
Finally, in Sec.~\ref{sec-consequences}, we 
study the time-energy resource needed to erase information in both the quantum and classical settings.
We conclude that
$\sqrt{(n+1)/n}$ 
times more resource is required in the quantum setting than in the classical setting
and
that
the amount of 
time-energy 
resource needed for $k$ runs of the depolarizing channel scales as $\sqrt{k}$ when the noise is small.


\section{Problem formulation}
\label{sec-problem-formulation}

\subsection{Notations and assumptions}

We can describe $\mathcal{F}(\rho)$ using the Kraus operators: $\mathcal{F}(\rho)=\sum_{j=1}^d F_j \rho F_j^\dag$ where $F_j \in \IC^{n\times n}$ are the Kraus operators satisfying the trace-preserving condition $\sum_{j=1}^d F_j^\dag F_j = I$.
Note that any channel can be described by at most $n^2$ Kraus operators.
But here the formulation is general for any number of Kraus operators.

Denote by $\myUgrp(r)$ the group of $r \times r$ unitary matrices.

Decompose 
$U \in \myUgrp(r)$ 
into eigenvectors:
\begin{equation}
\label{eqn-U-eigen-decomposition}
U=\sum_{j=1}^r
\exp(-i \theta_j)
\ket{u_j}\bra{u_j}
\end{equation}
where $\theta_j=E_j t/\hbar$,
$E_j$ is the energy, and $t$ is the evolution time.
We call $\theta_j$'s {\em eigenangles}.
We assume that all angles are taken in the range $(-\pi,\pi]$.
Define a time-energy measure for $U$~\cite{Chau2011}
\begin{align}
\munorm{U}&=\sum_{j=1}^r \mu_j |\theta_j|^\downarrow
\end{align}
where $|\theta_j|^\downarrow$ denotes $|\theta_j|$ ordered non-increasingly
$|\theta_1|^\downarrow \geq |\theta_2|^\downarrow \geq \cdots \geq |\theta_r|^\downarrow$. Also,
$\vec{\mu}=[\mu_1,\mu_2,\dots,\mu_r]\neq \vec{0}$
with 
$\mu_1 \geq \mu_2 \geq \cdots \geq \mu_r \geq 0$.
Note that $\munorm{U}$ satisfies the multiplicative triangle inequality $\munorm{UV} \leq \munorm{U}+\munorm{V}$~\cite{Chau2011}.

We have two special cases for the time-energy measure:

\hfill\begin{tabular}{l@{\hspace{.2cm}}l@{\hspace{.3cm}}r}
\text{$\bullet$ Sum time-energy:} & $\lVert U \rVert_\text{sum} \equiv \displaystyle\sum_{j=1}^{r} |\theta_j|$ . 
& $(\refstepcounter{equation}\theequation)\label{eqn-def-sum-energy}$
\\
\text{$\bullet$ Max time-energy:} & 
$\lVert U \rVert_\text{max} \equiv \displaystyle\max_{1 \leq j \leq r} |\theta_j|=|\theta_1|^\downarrow$.
& $ (\refstepcounter{equation}\theequation)\label{eqn-def-max-energy}$
\end{tabular}

\medskip

\noindent Note that the subscript ``sum'' is short for $\vec{\mu}=[1,1,\dots,1]$ and ``max'' for $\vec{\mu}=[1,0,\dots,0]$.

Define $\lVert \mathbf{v} \rVert=\sqrt{\sum_{j=1}^d |v_j|^2}$ where $\mathbf{v}=[v_1,v_2,\dots,v_d]$.

We adopt the convention that $\cos^{-1}$ always returns an angle in the range $[0,\pi]$.

\subsection{The ``partial $U$ problem'' and the ``channel problem''}

We generalize the measure $\munorm{\cdot}$ to the case where part of a unitary matrix is given.
Suppose that we are given the first $n \leq r$ columns of $U$ 
denoted by 
$U_{[1,n]} \in \IC^{r \times n}$.
Define the ``partial $U$ problem'' for $U_{[1,n]}$:
\dmathX2{
\munorm{U_{[1,n]}} &\equiv& \displaystyle\min_{V} & \munorm{V}\cr
&&\text{s.t.}&V_{[1,n]}=U_{[1,n]}\text{ and}\cr
&&&V \in \myUgrp(r).&eqn-energy-measure-partial-U\cr
}
We use this as a bridge to generalize the measure to quantum channel
$\mathcal{F}(\rho_A)$
acting on density matrix $\rho_A \in \IC^{n\times n}$.
For such a channel
$\mathcal{F}(\rho_A)=\sum_{i=1}^{d} F_i \rho_A F_i^\dag$ where $F_i \in \IC^{n\times n}$,
there exist unitary extensions $U_{BA}$ in a larger Hilbert space with an ancillary system $B$ such that 
$\mathcal{F}(\rho_A) = \tr_B [ U_{BA} (\ket{0}_B\bra{0} \otimes  \rho_A) U_{BA}^\dag ]$.

Define a mapping from a sequence of Kraus operators $F_{1:\KrausDim} \triangleq (F_1,F_2,\ldots, F_\KrausDim)$ to an $\KrausDim \ChannelDim \times \ChannelDim$ matrix 
as follows:
\begin{align}
g(F_{1:\KrausDim})
&\triangleq
\begin{bmatrix}
F_1 
\\
F_2 
\\
\vdots
\\
F_{\KrausDim}
\end{bmatrix}.
\end{align}
Because $\sum_{j=1}^{\KrausDim} F_j^\dag F_j = I$, the columns of $g(F_{1:\KrausDim})$ are orthonormal and $g(F_{1:\KrausDim})$ can be regarded as a submatrix of a unitary one.
Thus, we can obtain $\munorm{g(F_{1:\KrausDim})}$ from problem~\eqref{eqn-energy-measure-partial-U}.

Note that two sets of Kraus operators $\{F_1,\ldots, F_d\}$ and $\{F_1',\ldots, F_d'\}$ represent the same quantum channel if and only if $F_i'=\sum_{j=1}^d w_{ij} F_j$ for all $i$ and for some unitary matrix $[w_{ij}]$ (see Ref.~\cite{Nielsen2000}).
If more Kraus operators are desired in one set, we can supplement the other set with all-zero Kraus operators.
This implies that given one Kraus representation $\{F_1,\ldots, F_d\}$,
the most general form of all unitary extension implementing $\mathcal{F}$ is
\begin{align}
\label{eqn-U-general-form}
U_{BA}=(W_B \otimes I_A) 
\underbrace{
\begin{bmatrix}
F_1 & * & * & \cdots & *
\\
F_2 & * & * & \cdots & *
\\
\vdots & & & & \vdots
\\
F_{d} & * & * & \cdots & *
\\
\mathbf{0} & * & * & \cdots & *
\\
\vdots & & & & \vdots
\\
\mathbf{0} & * & * & \cdots & *
\end{bmatrix}
}_{\displaystyle\tilde{U}_{BA}}
\end{align}
where $W_B$ is any unitary of dimension $d' \times d'$, $I_A$ is the identity matrix of dimension $n \times n$, $\tilde{U}_{BA}$ is any unitary of dimension $nd' \times nd'$ with the first $n$ columns fixed as shown, and $\mathbf{0}$ is the all-zero matrix of dimension $n \times n$.
Here, we allow $d'$ to be in the range $d \leq d' < \infty$.

We define the time-energy measure of the quantum channel $\mathcal{F}$ given a Kraus representation $\{F_1,\ldots, F_d\}$ as follows:
\dmathX2{
\munorm{\mathcal{F}} &\equiv& \min_{W_B,d'} & \munorm{ (W_B \otimes I_A) \tilde{U}_{BA}(1:n) }\cr
&&\text{s.t.}&W_B \in \myUgrp(d')\text{ and }d \leq d' < \infty,&eqn-energy-measure-channel\cr
}
where we make use of Eq.~\eqref{eqn-energy-measure-partial-U} in the objective function and
$\tilde{U}_{BA}(1:n)$ are the first $n$ columns of $\tilde{U}_{BA}$ given in Eq.~\eqref{eqn-U-general-form}.
We call this the ``channel problem''.

Our ultimate goal is to find the minimum time-energy required to implement a quantum channel by solving the ``channel problem''~\eqref{eqn-energy-measure-channel}.
We first consider the ``partial $U$ problem''~\eqref{eqn-energy-measure-partial-U}.

\section{Simplification of the ``partial $U$ problem''}
\label{sec-simplification}

The ``partial $U$ problem''~(\ref{eqn-energy-measure-partial-U}) can be recast as that given $U$ of dimension $r \times r$ of the form
\begin{align}
\label{eqn-U-with-missing-columns}
U=
\begin{bmatrix}
\ket{b_1} & \ket{b_2} & \dots & \ket{b_n} 
& * & * & \cdots & *
\end{bmatrix}
\end{align}
where the first $n$ columns, labeled as $\ket{b_i}, i=1,\dots,n$, are fixed,
our goal is to find the remaining $r-n$ columns to 
minimize 
$\munorm{U}$
while maintaining $U$ unitary.

It is helpful to consider $U$ as a mapping with the requirement that it performs the following transformations:
\begin{align}
\label{eqn-overall-required-transformation}
\ket{e_i} \longrightarrow \ket{b_i}  \:\: \text{for all }i=1,\ldots,n
\end{align}
where $\ket{e_i}$ is the unit vector with $1$ at the $i$th entry and $0$ everywhere else.
Then, the ``partial $U$ problem''~(\ref{eqn-energy-measure-partial-U}) for 
$U_{[1,n]}$ is
equivalent to
\dmathX2{
\munorm{U_{[1,n]}} &=& \displaystyle\min_{U} & 
\munorm{ U }\cr
&&\text{s.t.}&U\ket{e_i} = \ket{b_i}  \:\: \text{for all }i=1,\ldots,n,\cr
&&&\text{with }U \in \myUgrp(r).&eqn-problem-original-min-U\cr
}
Note that $\{ \ket{b_i}: i=1,\dots,n\}$ is an orthonormal set due to the trace-preserving property of quantum channels.

\begin{lemma}
\label{lemma-U-norm-unchanged-by-conjugation}
{\myrm
$\munorm{ g(F_1,F_2,\dots,F_{\KrausDim}) }=
\munorm{ g(QF_1Q^\dag,F_2Q^\dag,\dots,F_{\KrausDim}Q^\dag) }
$
for any unitary matrix $Q$.
}
\end{lemma}
\begin{proof}
Let $G_1=g(F_1,F_2,\dots,F_{\KrausDim})$ and $G_2=g(QF_1Q^\dag,F_2Q^\dag,\dots,F_{\KrausDim}Q^\dag)$.
First note that
$
G_2=
\tilde{Q}
G_1
Q^\dag
$, where
$
\tilde{Q}=
\begin{bmatrix}
Q & 0
\\
0 & I
\end{bmatrix}
$.
Problem~(\ref{eqn-energy-measure-partial-U}) for $G_2$ is
\dmathX2{
\munorm{G_2} &\equiv& \min_{V} & \munorm{V}\cr
&&\text{s.t.}&V_{[1,n]}=G_2 \text{ and}\cr
&&&V \text{ is unitary}.
}
Pre- and post-multiplication on the constraint gives
\dmathX2{
\munorm{G_2} &\equiv& \min_{V} & \munorm{V}\cr
&&\text{s.t.}&\tilde{Q}^\dag V_{[1,n]} Q =
\tilde{Q}^\dag G_2 Q \text{ and}\cr
&&&V \text{ is unitary}.
}
Further simplification on the constraint gives
\dmathX2{
\munorm{G_2} &\equiv& \min_{V} & \munorm{\tilde{Q}^\dag V \tilde{Q}}\cr
&&\text{s.t.}&(\tilde{Q}^\dag V \tilde{Q})_{[1,n]}  =
G_1 \text{ and}\cr
&&&\tilde{Q}^\dag V \tilde{Q} \text{ is unitary}.
}
Here, we used the fact that $\munorm{V}=\munorm{\tilde{Q}^\dag V \tilde{Q}}$ for any unitary $\tilde{Q}$ since the eigenvalues are preserved under the conjugation by $\tilde{Q}$.
Finally, noting that minimizing over $V$ is the same as minimizing over $\tilde{Q}^\dag V \tilde{Q}$, the claim that $\munorm{G_1}=\munorm{G_2}$ is proved.
\end{proof}

\begin{remark}
\label{remark-triangular-F1}
{\myrm
(Triangularization of $F_1$)
According to Lemma~\ref{lemma-U-norm-unchanged-by-conjugation}, 
$\munorm{ g(QF_1Q^\dag,F_2Q^\dag,\dots,F_{\KrausDim}Q^\dag) }$ is invariant to unitary $Q$.
Thus, we may choose any $Q$ so that the Kraus operators are in a form that we desire.
In particular, we may choose $Q$ to be the unitary matrix of
the Schur decomposition of $F_1$.
(Note that the Schur decomposition is applicable to any matrix.)
This makes $QF_1Q^\dag$ upper triangular with the eigenvalues of $F_1$ on the diagonal.
}
\end{remark}

\section{``Partial $U$ problem'' with one vector}
\label{sec-single-vector-transformation}

We solve the ``partial $U$ problem''~\eqref{eqn-problem-original-min-U}
for the special case of $n=1$.
This case turns out to be useful in computing the upper bound of the time-energy $\munorm{\mathcal{F}}$ (Sec.~\ref{sec-upper-bound}), the lower bound of $\munorm{\mathcal{F}}$ (Sec.~\ref{sec-lower-bound}), and the optimal 
$\munorm{\mathcal{F}}$ for a class of quantum channels (Sec.~\ref{sec-optimal-energy}).

\subsection{Optimal single-vector transformation: general form}
\label{sec-optimal-single-vector-transformation}

We consider the optimal 
$U \in \myUgrp(\sizeofU)$ $(\sizeofU \geq 2)$ for this problem:
\dmathX2{
P_{\vec{\mu}}(\ket{a},\ket{b}) &\equiv& \min_{U} & \munorm{ U }\cr
&&\text{s.t.}&U\ket{a} = \ket{b}\text{ with}\cr
&&&U \in \myUgrp(r)&eqn-problem-min-U-single-vector\cr
}
where $\ket{a}$ and $\ket{b}$ are general normalized vectors of length $\sizeofU$.
We first show that the optimal $U$ can be achieved with two non-zero eigenangles for any $\vec{\mu}$.
Then we show how to construct $U$ given two eigenangles, find 
$\maxnorm{U}$, and bound $\sumnorm{U}$.

Note that the solution is of the form
\begin{align}
\label{eqn-problem-min-U-single-vector-solution-form}
P_{\vec{\mu}}(\ket{a},\ket{b})=
f_{\vec{\mu}}(\braket{a}{b}).
\end{align}
This is because 
$\munorm{ U }=\munorm{ V^\dag U V}$ for any unitary $V$ and thus
$P_{\vec{\mu}}(\ket{a},\ket{b})=
P_{\vec{\mu}}(V\ket{a},V\ket{b})$.

In the following, we assume $\ket{a} \neq \ket{b}$.
The case $\ket{a} = \ket{b}$ is trivial since 
$P_{\vec{\mu}}(\ket{a},\ket{a})=0$ with $U=I$.

\subsubsection{Optimal $U$ operates non-trivially on a two-dimensional subspace}

Consider the constraint $U\ket{a} = \ket{b}$.
We have
\begin{align}
\label{eqn-single-U-linear-combination}
\sum_{j=1}^{\sizeofU} \exp(i \theta_{j}) 
\lvert \braket{u_{j}}{a} \rvert^2
=\braket{a}{b}
,
\end{align}
where $U=\sum_{j=1}^{\sizeofU} \exp(i \theta_{j}) \ket{u_{j}}\bra{u_{j}}$ is the eigen-decomposition of $U$.
Thus, the point $\braket{a}{b}$ 
is a linear combination of the vertices $\exp(i \theta_{j})$ on the unit circle with weights $\lvert \braket{u_{j}}{a} \rvert^2$.
We show that the optimal $U$ for the time-energy measure $\munorm{U}$
can always be achieved with a linear combination of two vertices.

\begin{figure}
\begin{center}
\includegraphics[width=.5\columnwidth]{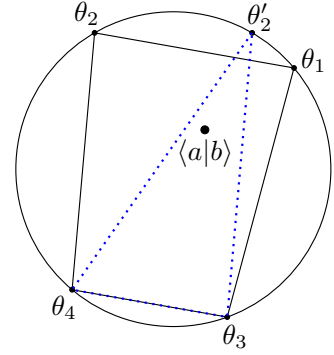}
\caption{
\label{fig:polygon}
The point $\braket{a}{b}$ is initially obtained as a linear combination of four points at eigenangles $\theta_j, j=1,\dots,4$ of $U$ based on Eq.~\eqref{eqn-single-U-linear-combination}.
A new eigenangle $\theta_2'$ can be found such that $\braket{a}{b}$ is a linear combination of $\theta_2'$, $\theta_3$, and $\theta_4$.
}
\end{center}
\end{figure}

\begin{lemma}
\label{lemma-Ui-two-terms}
{\myrm
$U$ with the minimal $\munorm{U}$
such that Eq.~\eqref{eqn-single-U-linear-combination} is satisfied
can always be achieved 
with two non-zero eigenangles.
}
\end{lemma}
\begin{proof}
Suppose the number of $\theta_{j}$ with non-zero weights (i.e., $\braket{u_{j}}{a} \neq 0$) is $m>2$.
From these $m$ vertices, pick two adjacent vertices that are either both positive or both negative.
This can always be done since $m>2$.
Without loss of generality (wlog), denote these two vertices as
 $\exp(i \theta_{1})$ and $\exp(i \theta_{2})$, and the remaining vertices as
$\exp(i \theta_{3})$ to $\exp(i \theta_{m})$.
Since $\sum_{j=1}^m \lvert \braket{u_{j}}{a} \rvert^2 =1$, the point $\braket{a}{b}$ lies inside the polygon defined by the vertices $\exp(i \theta_{j}), j=1,\dots,m$, according to Eq.~\eqref{eqn-single-U-linear-combination}.
If we replace the edge connecting $\exp(i \theta_{1})$ and $\exp(i \theta_{2})$ by their arc on the unit circle, the resultant shape will be strictly larger and contain the original polygon.
This new shape can be expressed as 
$$
\bigcup_{\theta'_2 \in [\theta_{1},\theta_{2}]}
\text{Polygon}(\theta'_2,\theta_{3},\ldots,\theta_{m})
$$
where $\text{Polygon}$ denotes the polygon defined by the vertices given in the arguments.
Here, we assume $\theta_{1}\leq\theta_{2}$ wlog.
Since the point $\braket{a}{b}$ lies inside this shape, it must also lie inside one of the polygons each defined with $m-1$ vertices.
Therefore, the point can be obtained as a linear combination of 
$m-1$ vertices 
(see Fig.~\ref{fig:polygon}).
Essentially, we replace $\theta_{1}$ and $\theta_{2}$ by some $\theta'_2$ defining the relevant polygon.
It remains to verify that 
the time-energy measure
using these $m-1$ vertices is no larger than before.
Denote by $P(j), j=1,\ldots,\sizeofU$, the decreasing order of $|\theta_j|$ where $\{\theta_{m+1},\ldots,\theta_\sizeofU\}$ are the eigenangles of $U$ with zero weights (i.e., $\braket{u_{j}}{a} = 0$).
Denote by $P'(j), j=1,\ldots,\sizeofU$, the decreasing order of $|\theta'_j|$ where 
$\theta'_j=\theta_j$ for $j=3,\ldots,m$ and $\theta'_j=0$ for $j=1,m+1,\ldots,\sizeofU$.
($\theta_2'$ is defined above.)
We have
\begin{align*}
\munorm{U}
=&
\sum_{j=1}^\sizeofU \mu_{P(j)} |\theta_j|
\\
\geq&
\sum_{j=3}^m \mu_{P'(j)} |\theta'_j|+ \mu_{P'(2)} \max(|\theta_2|,|\theta_1|)+ 
\\
&\mu_{P'(1)} \min(|\theta_2|,|\theta_1|)+\sum_{j=m+1}^\sizeofU \mu_{P'(j)} |\theta_j|
\\
\geq&
\sum_{j=3}^m \mu_{P'(j)} |\theta'_j|+\mu_{P'(2)} |\theta'_2|,
\end{align*}
where the second last line is due to that 
$\sum_{j=1}^\sizeofU \mu_{P(j)} |\theta_j| \geq \sum_{j=1}^\sizeofU \mu_{P''(j)} |\theta_j|$ for any ordering $P''$,
and the last line is due to that $\theta'_2 \in [\theta_{1},\theta_{2}]$.
In summary, a new $U'$ can be formed using these $m-1$ eigenangles $\{\theta'_2,\ldots,\theta'_m\}$ with  $\munorm{U'} \leq \munorm{U}$.
We can repeat this argument for removing another vertex until we reach $m=2$.
This proves that 
the optimal $U$ for the time-energy measure can always be achieved with a linear combination of two vertices or, in other words, two non-zero eigenangles.
\end{proof}

This lemma implies that when finding an optimal $U$ with respect to $\munorm{U}$,
it is sufficient to consider all chords (i.e., two-vertex polygons) on the unit circle passing through the desired point $\braket{a}{b}$.
Each chord defines two eigenangles, $\theta_1$ and $\theta_2$, which
in turn define a 
unitary transformation from $\ket{a}$ to $\ket{b}$.
This transformation $\tilde{U}$ acts on the subspace spanned by $\ket{a}$ and $\ket{b}$:
\begin{align}
\tilde{U}
&=
\tilde{u}_{1} \ket{a_i}\bra{a_i} +\tilde{u}_{2} \ket{a_i}\bra{a_i^\perp} 
+\tilde{u}_{3}\ket{a_i^\perp}\bra{a_i} + \tilde{u}_{4} \ket{a_i^\perp}\bra{a_i^\perp} 
\nonumber
\\
&=
\begin{bmatrix}
\tilde{u}_{1} & \tilde{u}_{2}
\\
\tilde{u}_{3} & \tilde{u}_{4}
\end{bmatrix}
\label{eqn-tilde-Ui}
\end{align}
expressed in the basis $\{\ket{a},\ket{a^\perp}\}$.
Here,
\begin{align}
\label{eqn-a-perp}
\ket{a^\perp}
&=
\frac{1}{\sqrt{1-\lvert \braket{a}{b} \rvert^2}}
\Big(\ket{b} - \braket{a}{b} \ket{a}\Big)
\end{align}
is a vector orthogonal to $\ket{a}$ in the plane spanned by $\ket{a}$ and $\ket{b}$.
(We assume $\ket{a} \neq \ket{b}$.)
The entries of $\tilde{U}$ can be found by imposing that the eigenvalues are $\exp(i\theta_1)$ and $\exp(i\theta_2)$ and $\tilde{U}\ket{a}=\ket{b}$ (see Appendix~\ref{app-elements-of-U} for detail):
\begin{align}
\label{eqn-U-with-actual-elements}
\tilde{U}=
\begin{bmatrix}
\braket{a}{b} & -e^{i (\theta_1+\theta_2)} \sqrt{1-|\braket{a}{b}|^2}
\\
\sqrt{1-|\braket{a}{b}|^2} & e^{i (\theta_1+\theta_2)} \braket{b}{a}
\end{bmatrix}.
\end{align}

The overall transformation is composed of the transformation $\tilde{U}$ in the subspace spanned by $\ket{a}$ and $\ket{b}$ and a transformation $\tilde{U}^\perp$ in the orthogonal subspace:
\begin{align}
\label{eqn-Ui-breakdown}
U
=
\tilde{U}
+
\tilde{U}^\perp.
\end{align}
We assign $\tilde{U}^\perp$ with zero eigenangles:
\begin{align}
\label{eqn-tilde-Uiperp}
\tilde{U}^\perp=
I-\ket{a}\bra{a}-\ket{a^\perp}\bra{a^\perp} .
\end{align}
This ensures that $\munorm{U}$ is minimized.
Thus, 
$
\munorm{U}=\munorm{\tilde{U}}
=\mu_{1} |\theta_i|+\mu_{2} |\theta_{3-i}|
$ where $i=\arg \max_{j=1,2} |\theta_j|$.
Note that the eigenvectors corresponding to $\tilde{U}^\perp$ in
Eq.~\eqref{eqn-single-U-linear-combination} have weights $\lvert \braket{u_{j}}{a} \rvert^2=0$.

\subsection{Optimal max time-energy}

\begin{figure}
\begin{center}
\includegraphics[width=.5\columnwidth]{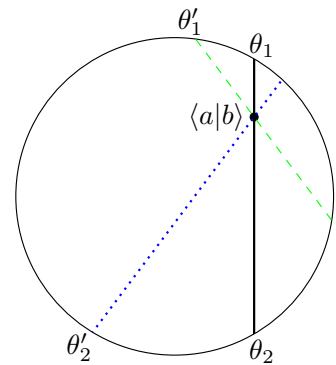}
\caption{
\label{fig:optimal-max-time-energy}
The optimal $U$ for the max time-energy consists of two non-trivial eigenangles $\theta_1$ and $\theta_2$ which form a vertical line passing through the point $\braket{a}{b}$.
Two other rotated lines about the point are shown.
The dashed line (green) and the dotted line (blue) have max time-energy $\theta_1'$ and $|\theta_2'|$ respectively, which are both greater than the optimal value $\theta_1$.
}
\end{center}
\end{figure}

For the max time-energy,
we show that the optimal $U$ of problem~\eqref{eqn-problem-min-U-single-vector} has two non-trivial eigenangles $\theta_1,\theta_2=\pm\cos^{-1}(\operatorname{Re}(\braket{a}{b}))$, where $\cos^{-1}$ always returns an angle in the range $[0,\pi]$,
and $U$ has the form of Eq.~\eqref{eqn-Ui-breakdown}.
The linear combination of these two eigenangles corresponds to a vertical line passing through the point $\braket{a}{b}$ (see Fig.~\ref{fig:optimal-max-time-energy}).
It can easily be seen that this line gives the minimal max time-energy.
Consider a line obtained by rotating the vertical line about $\braket{a}{b}$.
If $\braket{a}{b}$ is strictly inside the unit circle, then one of the two eigenangles must become larger in magnitude, giving rise to a larger max time-energy 
of
$\max(|\theta_1|,|\theta_2|)$.
If $\braket{a}{b}$ is on the unit circle, then one of the two eigenangles remains unchanged and so the max time-energy cannot become smaller.
Therefore, we have
\begin{align}
\label{eqn-max-norm-single-vector-solution}
P_{\text{max}}(\ket{a},\ket{b})
=
f_\text{max}(\braket{a}{b})
=
\cos^{-1}(\operatorname{Re}(\braket{a}{b})).
\end{align}

\subsection{Sum time-energy}

We only derive lower and upper bounds on the sum time-energy for problem~\eqref{eqn-problem-min-U-single-vector}.
\begin{lemma}
\label{lemma-minimum-angle-triangle}
{\myrm
For each chord that passes through the point $r \exp(i \gamma )$, we associate a triangle formed by the origin and the chord.
Among all such chords, the minimum angle of the triangle at the origin is $2\beta$ 
where $r = \cos(\beta)$.
}
\end{lemma}
\begin{proof}
Note that the problem is invariant to the rotation by $\gamma$.
Thus, wlog we assume $\gamma=0$.
Let the two end points of the chord be $\exp(i \zeta_1)$ and $\exp(-i \zeta_2)$, where $\zeta_1,\zeta_2 \geq 0$.
The angle in question is $\zeta_1+\zeta_2$ and we show that $\zeta_1+\zeta_2 \geq 2\beta$.
The point $r \exp(i 0 )$ is a linear combination of these two end points:
$r \exp(i 0 )= z \exp(i \zeta_1) + (1-z) \exp(-i \zeta_2)$, where $0 \leq z \leq 1$.
Thus, $\zeta_1$, $\zeta_2$, and $z$ have to satisfy the constraint on the magnitude:
\begin{align*}
r^2
&=
[z \cos(\zeta_1) + (1-z) \cos(\zeta_2)]^2+
\\
&\hspace{13pt}[z \sin(\zeta_1) - (1-z) \sin(\zeta_2)]^2
\\
&=
\left[2-2\cos(\zeta_1+\zeta_2)\right] (z^2 -z) + 1.
\end{align*}
This implies that $z$ is a function of $\zeta_1+\zeta_2$.
Solving the quadratic equation, we get
\begin{align*}
z=
\frac
{A \pm \sqrt{A^2 - 4 A (1-r^2)}}
{2A} ,
\end{align*}
where $A=2-2\cos(\zeta_1+\zeta_2)$.
Note that if $A=0$, then $\zeta_1=\zeta_2=0$ which implies that $r=1$ and $\beta=0$; thus, $0=\zeta_1+\zeta_2 \geq 2\beta=0$ as claimed.
Otherwise, $A>0$ and in this case, $z$ has a real solution if 
\begin{align*}
0 & \leq A-4(1-r^2)
\\
\Longrightarrow
\cos(\zeta_1+\zeta_2) & \leq 2r^2-1
= 2 \cos^2(\beta)-1 = \cos(2\beta)
\end{align*}
where $r = \cos(\beta)$.
Since $\cos$ is a decreasing function in the domain $[0,\pi]$,
$\zeta_1+\zeta_2 \geq 2\beta$ as claimed.
\end{proof}

\begin{remark}
\label{remark-minimum-triangle-angle}
{\myrm
Note that the minimum angle of $2\beta$ in Lemma~\ref{lemma-minimum-angle-triangle} is achieved by
the triangle formed by the origin and the chord perpendicular to the line connecting the origin and $r \exp(i \gamma )$.
}
\end{remark}

\begin{lemma}
{\myrm
The solution to problem~\eqref{eqn-problem-min-U-single-vector} for the sum time-energy is lower bounded as follows:
\begin{align}
P_{\text{{\rm sum}}}(\ket{a},\ket{b})
=
f_\text{{\rm sum}}(\braket{a}{b}) 
\geq& \: 2\cos^{-1}
|\braket{a}{b}|
\nonumber
\\
\equiv& \:
f_\text{{\rm sum}}^\text{L}(\braket{a}{b}) .
\label{eqn-sum-norm-single-vector-solution-lower-bound}
\end{align}
}
\end{lemma}
\begin{proof}
The sum time-energy of $U$ is 
$|\theta_1|+|\theta_2|$ since $U$ has only two non-trivial eigenangles due to Lemma~\ref{lemma-Ui-two-terms}.
The chord defined by these two eigenangles, $\theta_1$ and $\theta_2$, passes through the point $\braket{a}{b}=r \exp(i \gamma)$, where $r=|\braket{a}{b}|$.
We consider the triangle formed by the origin and the chord and focus on
the angle at the origin.
For the case $\theta_1>0$, $\theta_2\leq0$
and the case $\theta_1\leq0$, $\theta_2>0$,
this angle is $\min(|\theta_1-\theta_2|,2\pi-|\theta_1-\theta_2|)$. 
According to Lemma~\ref{lemma-minimum-angle-triangle}, this angle is lower bounded by $2\beta$
where $\beta = \cos^{-1}(r)$.
Thus,
\begin{align*}
|\theta_1|+|\theta_2| =|\theta_1-\theta_2| 
&\geq
\min(|\theta_1-\theta_2|,2\pi-|\theta_1-\theta_2|)
\\
&\geq 2\beta .
\end{align*}
Equality is achieved when $|\theta_1-\theta_2| \leq \pi$ and the chord described by $\theta_1$ and $\theta_2$ 
is perpendicular to the line connecting the origin and $r \exp(i \gamma)$ (see Remark~\ref{remark-minimum-triangle-angle}).
For the case 
$\theta_1,\theta_2>0$, the case 
$\theta_1,\theta_2<0$, 
and the case $\theta_1=\theta_2=0$, 
the angle 
is $|\theta_1-\theta_2|$ which is lower bounded by $2\beta$ according to Lemma~\ref{lemma-minimum-angle-triangle}. 
Thus, we have
\begin{align*}
|\theta_1|+|\theta_2| \geq |\theta_1-\theta_2| \geq 2\beta .
\end{align*}
\end{proof}

\begin{lemma}
{\myrm
The solution to problem~\eqref{eqn-problem-min-U-single-vector} for the sum time-energy is upper bounded as follows:
\begin{align}
P_{\text{{\rm sum}}}(\ket{a},\ket{b})
=
f_\text{{\rm sum}}(\braket{a}{b}) 
\leq& \:
2\cos^{-1}(\operatorname{Re}(\braket{a}{b}))
\nonumber
\\
\equiv& \:
f_\text{{\rm sum}}^\text{U}(\braket{a}{b}) .
\label{eqn-sum-norm-single-vector-solution-upper-bound}
\end{align}
}
\end{lemma}
\begin{proof}
Any $U$ that satisfies $U \ket{a}=\ket{b}$ (i.e., the constraint of problem~\eqref{eqn-problem-min-U-single-vector}) serves as an upper bound to 
$P_{\vec{\mu}}(\ket{a},\ket{b})$
for any $\vec{\mu}$.
Thus, for simplicity, we choose the optimal $U$ that achieves the optimal max time-energy in Eq.~\eqref{eqn-max-norm-single-vector-solution} to serve as an upper bound to 
$P_{\text{sum}}(\ket{a},\ket{b})$.
This $U$ has two non-trivial eigenangles $\theta_1=\cos^{-1}(\operatorname{Re}(\braket{a}{b}))$ and $\theta_2=-\theta_1$.
Thus, the sum time-energy of this $U$ is $|\theta_1|+|\theta_2|=2\cos^{-1}(\operatorname{Re}(\braket{a}{b}))$, and this is an upper bound to 
$P_{\text{sum}}(\ket{a},\ket{b})$.
\end{proof}

\section{Time-energy upper bound}
\label{sec-upper-bound}

In this section, we consider upper bounding $\munorm{\mathcal{F}}$ given its Kraus operators 
$(F_1,F_2,\dots,F_{\KrausDim})$ by upper bounding $\munorm{g(F_{1:\KrausDim})}$.
Any implementation $U$ of the quantum channel $\mathcal{F}$ serves as an upper bound to 
$\munorm{\mathcal{F}}$
since 
$\munorm{\mathcal{F}} \leq \munorm{U}$
for all $U$ of the form of
Eq.~\eqref{eqn-U-general-form}.
We propose a simple method to construct a time-energy efficient $U$ that completes the partial matrix $g(F_{1:\KrausDim})$, and $\munorm{U}$ will serve as an upper bound to the ``partial $U$ problem''~\eqref{eqn-energy-measure-partial-U} for $\munorm{g(F_{1:\KrausDim})}$,
which is an intermediate problem to the ultimate ``channel problem''~\eqref{eqn-energy-measure-channel}.

\subsection{Successive construction of $U$}

We focus on finding an upper bound to the ``partial $U$ problem''~\eqref{eqn-energy-measure-partial-U} which was recast as problem~\eqref{eqn-problem-original-min-U} which 
finds a unitary matrix $U$ that satisfies $n$ transformation rules:
$\ket{e_i} \longrightarrow \ket{b_i}  \:\: \text{for all }i=1,\ldots,n$.
Here, we focus on the ``partial $U$ problem'' for $\munorm{g(F_{1:\KrausDim})}$.
Thus, $\ket{b_i}$ is the $i$th column of $g(F_{1:\KrausDim})$.
In Sec.~\ref{sec-single-vector-transformation},
we analyzed the optimal unitary operation $U_i$ for each single-vector transformation $\ket{e_i} \longrightarrow \ket{b_i}$,
and we solved $\maxnorm{U_i}$ and bounded $\sumnorm{U_i}$.
Motivated by this result, we propose a greedy method to construct $U$ in which we successively construct the best unitary $U_i$ for each $i=1,\ldots,n$, and concatenate them.
The overall $U$ will be 
\begin{align}
\label{eqn-upper-bound-successive-U}
U=U_n \dots U_1.
\end{align}

We design each $U_i$ as follows.
When we consider the first transformation (i.e., $i=1$), we seek the optimal $U_1$ with the minimal $\munorm{U_1}$ such that
$$
U_1 \ket{e_1}= \ket{b_1}.
$$
When $i=2$, we seek the optimal $U_2$ such that
$$
U_2 U_1 \ket{e_2}= \ket{b_2}.
$$
In general, for the $i$th transformation, we seek the optimal $U_i$ such that
\begin{align}
\label{eqn-general-transform-1}
U_i \ket{a_i}= \ket{b_i}
\end{align}
where 
\begin{align}
\label{eqn-def-ai}
\ket{a_i}=U_{i-1}\cdots U_1 \ket{e_i} \text{ for } n \ge i\ge2 \text{ and } \ket{a_1}=\ket{e_1}.
\end{align}
Note that the optimal $U_i$ with the minimal $\munorm{U_i}$ has already been considered in problem~\eqref{eqn-problem-min-U-single-vector}.
We obtained the optimal value for $\maxnorm{U_i}$ and lower and upper bounds for $\sumnorm{U_i}$ (see Eqs.~\eqref{eqn-max-norm-single-vector-solution}, \eqref{eqn-sum-norm-single-vector-solution-lower-bound}, \eqref{eqn-sum-norm-single-vector-solution-upper-bound} ).
Thus,
\begin{align}
\label{eqn-Ui-function-of-ab}
\munorm{U_i}=f_{\vec{\mu}}(\braket{a_i}{b_i})
\end{align}
where the RHS comes from Eq.~\eqref{eqn-problem-min-U-single-vector-solution-form}.

A key feature of our construction is that we design $U_i$'s successively in a backward-looking fashion; i.e.,
when we design $U_i$, we only need to know $U_j$ for $j<i$ and we do not use $U_j$ for $j>i$.

In order for this successive approach to work, the action of a higher-index $U_{j}$ should not affect the transformation of a lower index $i<j$, i.e., 
\begin{align}
U_{i+1} U_i \ket{a_i} &= \ket{b_i}
\nonumber
\\
U_{i+2} U_{i+1} U_i \ket{a_i} &= \ket{b_i}
\nonumber
\\
&\vdots
\nonumber
\\
\label{eqn-subsequent-U-1}
U_n \cdots U_{i+2} U_{i+1} U_i \ket{a_i} &= \ket{b_i}.
\end{align}
Only if the last equation holds for all $i$ does the overall $U$ transforms according to Eq.~\eqref{eqn-overall-required-transformation} as required.
We show that Eq.~\eqref{eqn-subsequent-U-1} does hold.
\begin{lemma}
{\myrm
(Backward-looking design of $U_i$)
\begin{align}
U_{i+j} \cdots U_{i+2} U_{i+1} U_i \ket{a_i} &= \ket{b_i}
\end{align}
for $j\geq1$ when Eq.~\eqref{eqn-general-transform-1} holds.
}
\end{lemma}
\begin{proof}
We need to use the essential properties that $\braket{e_i}{e_k}=\braket{b_i}{b_k}=\delta_{ik}$ where $\delta_{ik}$ is the Kronecker delta.
We prove by induction.
First we compute $U_{i+1} U_i \ket{a_i}$.
Recall from Eq.~\eqref{eqn-Ui-breakdown} that $U_{i+1}$ performs a non-trivial transformation only in the subspace spanned by $\ket{a_{i+1}}$ and $\ket{b_{i+1}}$.
We show that $U_i \ket{a_i}$ is not in this subspace.
Note that 
\begin{align}
\label{eqn-lemma-higher-U-no-effect1}
\bra{a_{i+1}}{U_i \ket{a_i}}
&=
\left[
\bra{e_{i+1}} U_1^\dag U_2^\dag \cdots U_i^\dag
\right]
U_i
\Big[
U_{i-1} \cdots U_1 \ket{e_i}
\Big]
\\
&=
0
\nonumber
\end{align}
by
Eq.~\eqref{eqn-def-ai}, and
also
$\bra{b_{i+1}}{U_i \ket{a_i}}=\braket{b_{i+1}}{b_i}=0$ 
by
Eq.~\eqref{eqn-general-transform-1}.
Hence, $U_i \ket{a_i}$ is not in the aforementioned subspace.
This shows that 
\begin{align*}
U_{i+1} U_i \ket{a_i}
&=(\tilde{U}_{i+1}
+
\tilde{U}_{i+1}^\perp)
U_i \ket{a_i}
\\
&=
\tilde{U}_{i+1}^\perp
U_i \ket{a_i}
\\
&=
U_i \ket{a_i}
\\
&=
\ket{b_i}
\end{align*}
since $\tilde{U}_{i+1}^\perp$ acts trivially on the orthogonal complement of 
the subspace spanned by $\ket{a_{i+1}}$ and $\ket{b_{i+1}}$ (c.f. Eq.~\eqref{eqn-tilde-Uiperp}).

Now consider $U_{i+j} \cdots U_{i+1} U_i \ket{a_i}$ assuming the hypothesis $U_{i+j-1} \cdots U_{i+1} U_i \ket{a_i}=\ket{b_i}$.
Similar to Eq.~\eqref{eqn-lemma-higher-U-no-effect1}, 
\begin{eqnarray*}
&&
\bra{a_{i+j}}{U_{i+j-1} \cdots U_{i+1} U_i \ket{a_i}}
\\
&=&
\left[
\bra{e_{i+j}} U_1^\dag U_2^\dag \cdots U_{i+j-1}^\dag
\right]
U_{i+j-1} \cdots 
\\&&
U_{i+1} U_i
\Big[
U_{i-1} \dots U_1 \ket{e_i}
\Big]
\\
&=&
0
\end{eqnarray*}
and
\begin{align*}
\bra{b_{i+j}}
{U_{i+j-1} \cdots U_{i+1} U_i \ket{a_i}}
=
\braket{b_{i+j}}{b_i}=0 
\end{align*}
using the hypothesis.

Therefore, ${U_{i+j-1} \cdots U_{i+1} U_i \ket{a_i}}$ is not in the subspace spanned by
\ket{a_{i+j}} and \ket{b_{i+j}} (the subspace that $U_{i+j}$ acts non-trivially), and thus
\begin{align*}
U_{i+j} U_{i+j-1} \cdots U_{i+1} U_i \ket{a_i}
=
U_{i+j-1} \cdots U_{i+1} U_i \ket{a_i}
=
\ket{b_i}.
\end{align*}
This proves the claim.
\end{proof}
This is the key lemma that allows us to compute an upper bound in a successive manner.
In essence, instead of considering the original problem of finding a unitary required to perform $n$ simultaneous transformations, we consider the problem of finding $n$ unitaries each required to perform one transformation. 
Note that we already solved the latter problem in Sec.~\ref{sec-optimal-single-vector-transformation}.
In particular, the max time-energy for a single-vector transformation is given in Eq.~\eqref{eqn-max-norm-single-vector-solution}.
Thus, according to Eq.~\eqref{eqn-upper-bound-successive-U}, 
\begin{align}
\munorm{U}&=\munorm{U_n \dots U_1}
\nonumber
\\
&\leq
\sum_{i=1}^n \munorm{U_i}
\nonumber
\\
\label{eqn-upper-bound-sum-Ui1}
&=
\sum_{i=1}^n f_{\vec{\mu}}(\braket{a_i}{b_i})
\end{align}
where the inequality is due to the triangle inequality of the norm $\munorm{\cdot}$ (see Theorem~2 of Ref.~\cite{Chau2011}) and the last line is due to Eq.~\eqref{eqn-Ui-function-of-ab}.

In light of Remark~\ref{remark-triangular-F1}, we can always assume that $F_1$ is initially given in upper triangular form.
In the following,
we use this property to further deduce a simple bound on $\munorm{U}$ and consequently $\munorm{\mathcal{F}}$.
We show that when $F_1$ is upper triangular (which can always be guaranteed by Remark~\ref{remark-triangular-F1}),
$\ket{a_i}=\ket{e_i}$ for all $1 \le i \le n$ in Eqs.~\eqref{eqn-general-transform-1}--\eqref{eqn-def-ai}.
Thus, $\munorm{U_i}$ in Eq.~\eqref{eqn-upper-bound-sum-Ui1} will only depend on 
$\braket{e_i}{b_i}$ which is the $i$th eigenvalue of $F_1$.

\begin{lemma}
\label{lemma-independent-design-of-Ui}
{\myrm
(Independent design of $U_i$)
When $F_1$ is upper triangular,
$\ket{a_i}=\ket{e_i}$ for all $1 \le i \le n$ in Eqs.~\eqref{eqn-general-transform-1}--\eqref{eqn-def-ai},
which implies that designing according to
\begin{align}
U_i \ket{e_i} = \ket{b_i}
\end{align}
and designing according to Eq.~\eqref{eqn-general-transform-1} are equivalent.
}
\end{lemma}
\begin{proof}

Note that since $F_1$ is upper triangular, $\braket{b_{j}}{e_i}=0$ for 
$j<i\le n$.
We prove by induction.
We show that $\ket{a_{i}}=\ket{e_{i}}$ assuming the hypothesis that $\ket{a_j}=\ket{e_j}$ for $j<i$ is true.
Note that $\ket{a_1}=\ket{e_1}$ by definition in Eq.~\eqref{eqn-def-ai}.
Recall from Eq.~\eqref{eqn-Ui-breakdown} that $U_{j}$ performs a non-trivial transformation only in the subspace spanned by $\ket{a_{j}}$ and $\ket{b_{j}}$.
We show that $\ket{e_i}$ for $i>j$ is not in this subspace:
$
\braket{a_{j}}{e_{i}}
=
\braket{e_{j}}{e_{i}}
=0
$
where we assume that the hypothesis is true,
and
$\braket{b_{j}}{e_{i}}=0$ due to the triangular structure of $F_1$.
This means that for $i>j$,
\begin{align*}
U_j \ket{e_i}
&=(\tilde{U}_{j}
+
\tilde{U}_{j}^\perp)
\ket{e_i}
\\
&=
\tilde{U}_{j}^\perp
\ket{e_i}
\\
&=
\ket{e_i}
\end{align*}
since $\tilde{U}_{j}^\perp$ acts trivially (c.f. Eq.~\eqref{eqn-tilde-Uiperp}).
This shows that
\begin{align*}
\ket{a_i}&=
U_{i-1}\cdots U_1 \ket{e_i}
\\
&=
\ket{e_i}.
\end{align*}
\end{proof}

This lemma allows us to compute the upper bound easier since 
the upper bound in Eq.~\eqref{eqn-upper-bound-sum-Ui1}
now becomes
\begin{align}
\munorm{U}
&\leq
\sum_{i=1}^n f_{\vec{\mu}}(\braket{e_i}{b_i})
\\
&=
\sum_{i=1}^n f_{\vec{\mu}}(\lambda_i(F_1))
\end{align}
where we recognize that the diagonal elements of $F_1$ are its eigenvalues, denoted by $\lambda_i(F_1)$.
This shows that $\munorm{U}$ only depends on the eigenvalues of one Kraus operator of the quantum channel.
Since $\munorm{\mathcal{F}} \leq \munorm{U}$ for all $U$ of the form of
Eq.~\eqref{eqn-U-general-form}, we have
\dmathX2{
\munorm{\mathcal{F}} &\leq& \min_\mathbf{v} & \sum_{i=1}^n f_{\vec{\mu}}(\lambda_i(\sum_{j=1}^d v_j F_j))\cr
&&\text{s.t.}&\sum_{j=1}^d |v_j|^2 \leq 1 .&eqn-upper-bound-general-form\cr
}
Here, $\mathbf{v}$ corresponds to the first $d$ elements of the first row of $W_B$ in Eq.~\eqref{eqn-U-general-form}.
This bound holds for any quantum channel $\mathcal{F}$ described by Kraus operators $F_i, i=1,\dots,d$.

\subsection{Max time-energy upper bound}

For the max time-energy, we substitute $f_\text{max}$ in Eq.~\eqref{eqn-max-norm-single-vector-solution} for $f_{\vec{\mu}}$ in Eq.~\eqref{eqn-upper-bound-general-form} to get
\begin{align}
\maxnorm{\mathcal{F}} &\leq 
\min_{\mathbf{v}: \: \lVert \mathbf{v} \rVert \leq 1}
\sum_{i=1}^n 
\cos^{-1}
\Big[
\operatorname{Re}
(\lambda_i(\sum_{j=1}^d v_j F_j)) 
\Big].
\end{align}

\subsection{Sum time-energy upper bound}

For the sum time-energy, we substitute $f_\text{sum}^\text{U}$ in Eq.~\eqref{eqn-sum-norm-single-vector-solution-upper-bound} for $f_{\vec{\mu}}$ in Eq.~\eqref{eqn-upper-bound-general-form} to get
\begin{align}
\label{eqn-sum-energy-upper-bound}
\sumnorm{\mathcal{F}} &\leq 
\min_{\mathbf{v}: \: \lVert \mathbf{v} \rVert \leq 1}
\sum_{i=1}^n 
2\cos^{-1}
\Big[
\operatorname{Re}
(\lambda_i(\sum_{j=1}^d v_j F_j)) 
\Big].
\end{align}


\subsection{Special case: diagonal $F_1$}

Here, we consider 
the special class of channels for which $F_1$ is diagonal, which will be useful for showing the optimal time-energy for a class of channels in Sec.~\ref{sec-optimal-energy}.
Following the construction of $U$ in Eq.~\eqref{eqn-upper-bound-successive-U},
we argue not only that we can independently design $U_i$ (Lemma~\ref{lemma-independent-design-of-Ui}), but also that they act on orthogonal subspaces (meaning that they commute).
According to Lemma~\ref{lemma-independent-design-of-Ui}, $U_i$ transforms $\ket{e_i}$ to $\ket{b_i}$, and
is the solution to problem~\eqref{eqn-problem-min-U-single-vector} where $\ket{a}=\ket{e_i}$ and $\ket{b}=\ket{b_i}$. 
Thus, its non-trivial part $\tilde{U}_i$ [see Eq.~\eqref{eqn-Ui-breakdown}] acts on the subspace spanned by $\{\ket{e_i},\ket{b_i}\}$.

The fact that $F_1$ is diagonal gives rise to the following property:
$\braket{e_i}{b_j}=0$ if $i\neq j$ where $1 \leq i,j \leq n$.
Also, we already have 
$\braket{e_i}{e_j}=0$ if $i\neq j$ 
by definition and
$\braket{b_i}{b_j}=0$ if $i\neq j$
by the trace-preserving property of quantum channels.
Thus, the subspaces spanned by $\{\ket{e_i},\ket{b_i}\}$ for $i=1,\dots,n$ are orthogonal to each other.
This means that $\tilde{U}_i$ acting on these subspaces [see Eq.~\eqref{eqn-tilde-Ui}] are orthogonal to each other.
Then we may bypass the construction of $U_i$ in Eq.~\eqref{eqn-Ui-breakdown} and directly form
$$
U=\sum_{i=1}^n \tilde{U}_i + P_{U^\perp}
$$
where $P_{U^\perp}$ is the projection onto the subspace orthogonal to the summation term.
Thus, the set of non-zero eigenangles of $U$ is composed of the eigenangles of $\tilde{U}_i$ for all $i=1,\dots,n$.
This means that
\begin{align}
\maxnorm{U}&=\max_{1\leq i \leq n} \maxnorm{\tilde{U}_i}
\\
\sumnorm{U}&=\sum_{i=1}^n \sumnorm{\tilde{U}_i} .
\end{align}
Thus, for this class of channels, by using Eqs.~\eqref{eqn-max-norm-single-vector-solution}, \eqref{eqn-sum-norm-single-vector-solution-upper-bound}, \eqref{eqn-Ui-function-of-ab} and $\munorm{U_i}=\munorm{\tilde{U}_i}$, we have
\begin{align}
\label{eqn-upper-bound-diagonal-E1-maxnorm}
\maxnorm{\mathcal{F}} &\leq \maxnorm{U} =
\max_{1\leq i \leq n}
\cos^{-1} [\operatorname{Re} (\lambda_i(F_1)) ]
\\
\sumnorm{\mathcal{F}} &\leq
\sumnorm{U}
\leq
\sum_{i=1}^n
2 \cos^{-1} [\operatorname{Re} (\lambda_i(F_1)) ] .
\end{align}

\section{Time-energy Lower bound}
\label{sec-lower-bound}

\subsection{General form of lower bound}
We lower bound $\munorm{\mathcal{F}}$ of the ``channel problem''~\eqref{eqn-energy-measure-channel}.
We first consider lower bounding
$\munorm{g(F_{1:\KrausDim})}$ for a fixed set of Kraus operators $(F_1,F_2,\dots,F_{\KrausDim})$.
Note that $\munorm{g(F_{1:\KrausDim})}$ is obtained from the ``partial $U$ problem''~\eqref{eqn-problem-original-min-U},
and we propose a modified problem whose solution lower bounds this problem.
The modified problem is formed by removing all except one transformation constraints in problem~\eqref{eqn-problem-original-min-U} as follows:
\dmathX2{
P_{\vec{\mu}}(\ket{e_i},\ket{b_i}) &=& \displaystyle\min_{U} & 
\munorm{ U }\cr
&&\text{s.t.}&U\ket{e_i} = \ket{b_i}\text{ with}\cr
&&&U \in \myUgrp(r)
}
which is defined for $i=1,\dots,n$.
Here $\ket{b_i}$ is the $i$th column of 
$g(QF_1Q^\dag,F_2Q^\dag,\dots,F_{\KrausDim}Q^\dag)$ where $QF_1Q^\dag$ is an upper triangular matrix corresponding to the Schur decomposition of $F_1$ (see Remark~\ref{remark-triangular-F1}).
Note that this problem 
is the single-vector problem~\eqref{eqn-problem-min-U-single-vector} which we analyzed in Sec.~\ref{sec-single-vector-transformation}.
Certainly, the feasible set of this problem contains that of problem~\eqref{eqn-problem-original-min-U}, and so with the help of Lemma~\ref{lemma-U-norm-unchanged-by-conjugation}, we have 
$\munorm{ g(F_1,F_2,\dots,F_{\KrausDim}) }=\munorm{g(QF_1Q^\dag,F_2Q^\dag,\dots,F_{\KrausDim}Q^\dag)} \geq 
P_{\vec{\mu}}(\ket{e_i},\ket{b_i})
$ for all $i=1,\dots,n$.
Thus, 
\begin{align}
\munorm{ g(F_1,F_2,\dots,F_{\KrausDim}) }
&\geq 
\max_{1\leq i \leq n}
P_{\vec{\mu}}(\ket{e_i},\ket{b_i})
\nonumber
\\
&=
\max_{1\leq i \leq n} f_{\vec{\mu}}(\braket{e_i}{b_i})
\nonumber
\\
&=
\max_{1\leq i \leq n} f_{\vec{\mu}}(\lambda_i(F_1))
\nonumber
\end{align}
where we used Eq.~\eqref{eqn-problem-min-U-single-vector-solution-form} in the second line and 
the third line is due to Remark~\ref{remark-triangular-F1} with $\lambda_i(F_1)$ being the $i$th eigenvalue of $F_1$.
Note that the last inequality holds when the LHS is replaced by $\munorm{ g(F_1,F_2,\dots,F_{\KrausDim},\mathbf{0},\dots,\mathbf{0}) }$ for any number of extra all-zero Kraus operators inserted.
Combining with 
the ``channel problem''~\eqref{eqn-energy-measure-channel} ,
we have
\dmathX2{
\munorm{\mathcal{F}} &\geq& \min_\mathbf{v} & \max_{1\leq i \leq n}
f_{\vec{\mu}}(\lambda_i(\sum_{j=1}^d v_j F_j))\cr
&&\text{s.t.}&\sum_{j=1}^d |v_j|^2 \leq 1 .&eqn-lower-bound-general-form\cr
}
Here, $\mathbf{v}$ corresponds to the first $d$ elements of the first row of $W_B$ in Eq.~\eqref{eqn-U-general-form}.
This bound holds for any quantum channel $\mathcal{F}$ described by Kraus operators $F_i, i=1,\dots,d$.

\subsection{Max time-energy lower bound}
For the max time-energy, we substitute $f_\text{max}$ in Eq.~\eqref{eqn-max-norm-single-vector-solution} for $f_{\vec{\mu}}$ in Eq.~\eqref{eqn-lower-bound-general-form} to get
\begin{align}
\label{eqn-max-energy-lower-bound1}
\maxnorm{\mathcal{F}} 
\geq
&
\min_{\mathbf{v}: \: \lVert \mathbf{v} \rVert \leq 1}
\:\:
\max_{1\leq i \leq n}
\cos^{-1} 
\Big[
\operatorname{Re}(\lambda_i(\sum_{j=1}^d v_j F_j)) 
\Big].
\end{align}

\subsection{Sum time-energy lower bound}
For the sum time-energy, we substitute $f_\text{sum}$ for $f_{\vec{\mu}}$ in Eq.~\eqref{eqn-lower-bound-general-form} 
to get
\begin{align}
\sumnorm{\mathcal{F}} \geq 
\min_{\mathbf{v}: \: \lVert \mathbf{v} \rVert \leq 1}
\:\:
\max_{1\leq i \leq n}
f_\text{sum}(\lambda_i(\sum_{j=1}^d v_j F_j)) 
\end{align}
and it is a simple argument to argue that the lower bound of $\sumnorm{\mathcal{F}}$ can be  given in terms of
$f_\text{sum}^\text{L}$ defined in Eq.~\eqref{eqn-sum-norm-single-vector-solution-lower-bound}.
First, note that
\begin{align*}
P(\mathbf{v}) &\equiv
\max_i f_\text{sum}(\lambda_i(\sum_{j=1}^d v_j F_j)) 
\\
&\geq 
\max_i f_\text{sum}^\text{L}(\lambda_i(\sum_{j=1}^d v_j F_j)) 
\equiv
Q(\mathbf{v})
\end{align*}
for all $\mathbf{v}$.
And we have
\begin{align*}
\min_{\mathbf{v}} P(\mathbf{v}) = P(\mathbf{\hat{v}}) \geq Q(\mathbf{\hat{v}}) \geq \min_{\mathbf{v}} Q(\mathbf{v}) 
\end{align*}
where 
$\mathbf{\hat{v}} =\arg \min_{\mathbf{v}} P(\mathbf{v})$.
Therefore, we have
\begin{align*}
\sumnorm{\mathcal{F}} 
&\geq 
\min_{\mathbf{v}: \: \lVert \mathbf{v} \rVert \leq 1}
Q(\mathbf{v})
\\
&=
\min_{\mathbf{v}: \: \lVert \mathbf{v} \rVert \leq 1}
\:\:
\max_{1\leq i \leq n}
2 \cos^{-1} |\lambda_i(\sum_{j=1}^d v_j F_j)| .
\end{align*}
Note that $Q(\mathbf{v}) \leq Q(\kappa\mathbf{v})$ for $0 \leq \kappa \leq 1$, and thus we have
\begin{align}
\sumnorm{\mathcal{F}} 
&\geq 
\min_{\mathbf{v}: \: \lVert \mathbf{v} \rVert = 1}
\:\:
\max_{1\leq i \leq n}
2 \cos^{-1} |\lambda_i(\sum_{j=1}^d v_j F_j)| .
\end{align}

\section{Optimal time-energy for a class of channels}
\label{sec-optimal-energy}

\begin{definition}
\label{def-class1}
{\myrm
Define a class $\mathcal{C}(n)$ of quantum channels acting on 
$n \times n$
density matrices 
where each channel is described by
Kraus operators $\{ F_j \in \IC^{n \times n}: j=1,\ldots,d \}$ of the form
\begin{align}
{F}_1 &= \sqrt{p} I \text{ where $0 \leq p \leq 1$}
\label{eqn-channel-class-E1}
\\
\tr({F}_j) &= 0, \:\: j=2,\ldots,d.
\nonumber
\end{align}
The number $d$ of Kraus operators of each channel can be different.
}
\end{definition}
Note that this class contains the depolarizing channel, the bit-flip channel, and the phase-flip channel.

\subsection{Optimal max time-energy}
We show that the lower bound of $\maxnorm{\mathcal{F}}$ in
Eq.~\eqref{eqn-max-energy-lower-bound1} is achievable for any
$\mathcal{F} \in \mathcal{C}(n)$ with $2 \leq n < \infty$.
The RHS of Eq.~\eqref{eqn-max-energy-lower-bound1} can be written as
\begin{align}
\label{eqn-max-energy-lower-bound2-swap}
\cos^{-1} 
\Big[
\max_{\mathbf{v}: \: \lVert \mathbf{v} \rVert \leq 1}
\:\:
\min_{1\leq i \leq n}
\operatorname{Re}(\lambda_i(\sum_{j=1}^d v_j F_j)) 
\Big],
\end{align}
since $\cos^{-1}$ is a decreasing function in the range $[0,\pi]$.
Consider part of this term:
\begin{align*}
P(\mathbf{v}) &\equiv
\min_{1\leq i \leq n} \operatorname{Re}(\lambda_i(\sum_{j=1}^d v_j F_j))
\\
& \leq 
\frac{1}{n}
\sum_{i=1}^n
\operatorname{Re}(\lambda_i(\sum_{j=1}^d v_j F_j))
\equiv Q(\mathbf{v}) .
\end{align*}
Let the eigenvalues of $\sum_{j=2}^d v_j F_j$ be $\{\sigma_1,\dots,\sigma_n\}$.
Note that $\sum_{i=1}^n \sigma_i = 0$.
Thus, 
$$
\lambda_i(\sum_{j=1}^d v_j F_j)
=
v_1 \sqrt{p} + \sigma_i
$$
and
$$
Q(\mathbf{v}) =
\frac{1}{n}
\sum_{i=1}^n
\operatorname{Re} (v_1 \sqrt{p} + \sigma_i)
=
\operatorname{Re} (v_1 \sqrt{p})
\leq \sqrt{p} .
$$
Considering the maximization in 
Eq.~\eqref{eqn-max-energy-lower-bound2-swap},
$$
\max_{\mathbf{v}} P(\mathbf{v})
\leq Q(\mathbf{\hat{v}})
\leq \max_{\mathbf{v}}
Q(\mathbf{v})
$$
where $\mathbf{\hat{v}}$ is the optimal value of the maximization of $P$.
Thus,
$$
\maxnorm{\mathcal{F}} \geq
\cos^{-1} (
\max_{\mathbf{v}} P(\mathbf{v})
)
\geq
\cos^{-1} (
\sqrt{p}
)
$$
since $\cos^{-1}$ is a decreasing function in the range $[0,\pi]$.
Note that the RHS 
coincides with the upper bound from
Eq.~\eqref{eqn-upper-bound-diagonal-E1-maxnorm}
where $\lambda_i(F_1)=\sqrt{p}$.
Therefore,
\begin{equation}
\label{eqn-channel-class-optimal-maxnorm}
\maxnorm{\mathcal{F}} =
\cos^{-1} (
\sqrt{p}
) 
\end{equation}
for any
$\mathcal{F} \in \mathcal{C}(n)$ with $2 \leq n < \infty$.

\section{Some interesting consequences}
\label{sec-consequences}

\subsection{Comparison of quantum and classical noisy channels}

\newcommand{\sigmaS}{S}

The quantum noisy channel or the quantum depolarizing channel acting on 
$n \times n$
density matrices is defined as
\begin{align*}
\mathcal{F}_Q (\rho) & \triangleq q \rho + (1-q) \frac{I}{n}
\end{align*}
where 
complete positivity requires that
$-1/(n^2-1) \leq q \leq 1$~\cite{King2003}.

The Weyl operators are the $n$-dimensional generalization of the Pauli operators and are defined as 
$$
\sigmaS_{jk}=\sum_{s=0}^{n-1} \omega^{sk} \ket{s+j}\bra{s} \in \IC^{n \times n}
$$
where $j,k=0,\dots,n-1$ and $\omega$ is the $n$th root of unity.
These operators have the following properties:
\begin{itemize}
\item
(Identity) $S_{00}=I$.
\item
(Traceless) $\tr (S_{jk})=0$ for $(j,k)\neq(0,0)$.
\item
(Complete erasure)
$
n^{-2}
\sum_{j,k=0}^{n-1} \sigmaS_{jk} \rho \sigmaS_{jk}^\dag = n^{-1} I
$
for any density matrix $\rho$.
\item
(Trace preserving)
$
n^{-2}
\sum_{j,k=0}^{n-1} \sigmaS_{jk}^\dag \sigmaS_{jk} =  I
$.
\item
(Complete erasure)
$
n^{-1}
\sum_{j=0}^{n-1} \sigmaS_{j0} \rho \sigmaS_{j0}^\dag = n^{-1} I
$
for any diagonal density matrix $\rho$.
\item
(Trace preserving)
$
n^{-1}
\sum_{j}^{n-1} \sigmaS_{j0}^\dag \sigmaS_{j0} = I
$.
\end{itemize}

We may express the quantum noisy channel 
using the Weyl's operators:
\begin{align}
\mathcal{F}_Q (\rho) 
&=
q I \rho I + (1-q) \frac{1}{n^2} \sum_{j,k=0}^{n-1} \sigmaS_{jk} \rho \sigmaS_{jk}^\dag .
\label{eqn-quantum-erasure-channel1}
\end{align}
The advantage of doing so is that we can now see that $\mathcal{F}_Q (\rho)$
is in the class $\mathcal{C}(n)$ (defined in Definition~\ref{def-class1}) and we can apply
Eq.~\eqref{eqn-channel-class-optimal-maxnorm} to get the max time-energy of it.

In a similar manner, we define the classical noisy channel which adds classical noise (i.e., classical states are being remapped):
\begin{align}
\mathcal{F}_C (\rho) 
&\triangleq
q I \rho I + (1-q) \frac{1}{n} \sum_{j=0}^{n-1} \sigmaS_{j0} \rho \sigmaS_{j0}^\dag ,
\label{eqn-classical-erasure-channel1}
\end{align}
where $-1/(n-1) \leq q \leq 1$ for complete positivity.
When $\rho$ is diagonal (i.e., a mixture of classical states), 
$$\mathcal{E}_C (\rho) 
 = q \rho + (1-q) \frac{I}{n}.$$

We now verify that we have a fair comparison between the quantum and classical noisy channels,
by checking that
the same amount of noise is added to the input states of both channels.
To quantify this, we use the trace distance to measure the difference between the input state and the output state, and we take the input state to be a pure state.
For the quantum depolarizing channel, consider any pure input state $\ket{\Phi}\bra{\Phi}$ and the trace distance is
\begin{align}
\frac{1}{2}
\Big\lVert \ket{\Phi}\bra{\Phi} - \mathcal{F}_Q (\ket{\Phi}\bra{\Phi}) \Big\rVert_\text{tr}
&=
\frac{1}{2}
(1-q)\Big\lVert \ket{\Phi}\bra{\Phi} - \frac{I}{n} \Big\rVert_\text{tr}
\nonumber
\\
&=
\frac{(1-q) (n-1)}{n} \equiv \delta
\label{eqn-erasure-channel-trace-distance}
\end{align}
where $\lVert A \rVert_\text{tr}$ denotes the trace norm of $A$ and is equal to the sum of all singular values of $A$.
For the classical noisy channel, we only consider classical states and so the input state is $\ket{j}, j=0,\dots,n-1$, and in this case, the trace distance is
\begin{align*}
\frac{1}{2}
\Big\lVert \ket{j}\bra{j} - \mathcal{F}_C (\ket{j}\bra{j}) \Big\rVert_\text{tr}
&=
\frac{1}{2}
(1-q)\Big\lVert \ket{j}\bra{j} - \frac{I}{n} \Big\rVert_\text{tr}
\\
&=
\delta .
\end{align*}
This shows that both $\mathcal{F}_Q$ and $\mathcal{F}_C$ 
adds the same amount of noise when both are characterized by the same parameter $q$.
Note that the two channels have different valid ranges of $q$, but this does not affect our discussion since we will focus on $q \approx 1$.

It can be easily seen that both $\mathcal{F}_Q$ and $\mathcal{F}_C$ written with the Weyl's operators are in the class $\mathcal{C}(n)$.
In this case, $\maxnorm{\mathcal{F}_Q}$ and $\maxnorm{\mathcal{F}_C}$ only depend on the respective scaling factors of the identity Kraus operators (c.f. Eq.~\eqref{eqn-channel-class-optimal-maxnorm}).
Note that $S_{00}=I$ and thus from
Eqs.~\eqref{eqn-quantum-erasure-channel1} and \eqref{eqn-classical-erasure-channel1}, the identity Kraus operators (c.f. Eq.~\eqref{eqn-channel-class-E1}) are
\begin{align*}
F_{Q1}&=\sqrt{ q + \frac{1-q}{n^2} } I 
=
\sqrt{1-\delta \left(\frac{n+1}{n}\right)} I
\\
F_{C1}&=\sqrt{ q + \frac{1-q}{n} } I 
=
\sqrt{1-\delta} I
\end{align*}
where we have used Eq.~\eqref{eqn-erasure-channel-trace-distance}.
Note that when $q=1$, $F_{Q1}=F_{C1}=I$; when $q=-1/(n^2-1)$, $F_{Q1}=0$ and when $q=-1/(n-1)$, $F_{C1}=0$.
Using these in Eq.~\eqref{eqn-channel-class-optimal-maxnorm}, we have
\begin{align}
\maxnorm{\mathcal{F}_Q}
&=\cos^{-1}
\sqrt{1-\delta \left(\frac{n+1}{n}\right)}
\\
\maxnorm{\mathcal{F}_C}
&=\cos^{-1}
\sqrt{1-\delta}
\end{align}
where in both cases the distance between the input and output states is $\delta$.
When $\delta \approx 0$, 
using the approximation
$\cos^{-1} \sqrt{x} \approx \sqrt{1-x}$ for $x \approx 1$, 
it can be shown that
\begin{align}
\maxnorm{\mathcal{F}_Q}
&=
\sqrt{\frac{n+1}{n}}
\maxnorm{\mathcal{F}_C} .
\end{align}
This shows that it takes 
$\sqrt{(n+1)/n}$ 
times more time-energy resource for a quantum process to erase information of the input state by the same distance $\delta$ compared to a classical process.
For two-level systems ($n=2$), this is $1.22$ times larger.

\subsection{Cascade of 
depolarizing channels with small noise
}

As discussed in the last subsection, the quantum depolarizing channel
\begin{align*}
\mathcal{F}_Q (\rho) & \triangleq q \rho + (1-q) \frac{I}{n}
\end{align*}
is in the class $\mathcal{C}(n)$ (defined in Definition~\ref{def-class1}) and we can apply
Eq.~\eqref{eqn-channel-class-optimal-maxnorm} to get
\begin{align}
\maxnorm{\mathcal{F}_Q}
&=\cos^{-1}
\sqrt{ q + \frac{1-q}{n^2} } .
\end{align}
Suppose that we run this channel 
$k$ times.
We can either run (i) a unitary implementation of $\mathcal{F}_Q$ $k$ times (with the ancilla reset before the start of a new run) or 
(ii) a unitary implementation of 
$$
\mathcal{F}_Q^{(k)} (\rho) \triangleq
\underbrace{\mathcal{F}_Q \circ \mathcal{F}_Q \circ \cdots \circ \mathcal{F}_Q}_{\mbox{$k$ times}} (\rho)
= q^k \rho + (1-q^k) \frac{I}{n}.
$$
For case (i), 
since each run is executed independently of each other,
the total time-energy is $k \maxnorm{\mathcal{F}_Q}$.

For case (ii), 
applying Eq.~\eqref{eqn-channel-class-optimal-maxnorm} gives
$$
\maxnorm{\mathcal{F}_Q^{(k)}} 
=
\cos^{-1} \sqrt{
q^k + \frac{1-q^k}{n^2}
}.
$$
Using Taylor series expansion, we approximate $\cos^{-1} \sqrt{x} \approx \sqrt{1-x}$ for $x \approx 1$ and 
$$
q^k + \frac{1-q^k}{n^2}
\approx 1 - k \left( 1- \frac{1}{n^2} \right) (1-q)$$ 
for $q \approx 1$.
Thus, we have 
$$
\maxnorm{\mathcal{F}_Q^{(k)}} 
\approx 
\sqrt{k}
\sqrt{1- \left(q+\frac{1-q}{n^2}\right)}
$$
which implies that
\begin{align}
\maxnorm{\mathcal{F}_Q^{(k)}} 
\approx 
\sqrt{k}
\:
\maxnorm{\mathcal{F}_Q} .
\end{align}
This means 
that considering all $k$ channels together saves time-energy resource by a factor of $\sqrt{k}$ compared to separately running the channels when the noise is small ($q \approx 1$).

\section{Concluding remarks}

In this paper, we extend the time-energy measure proposed by Chau~\cite{Chau2011} to general quantum processes.
This measure is a good indicator of the time-energy tradeoff of a quantum process.
Essentially, a large time-energy value suggests that the quantum process takes a longer time or more energy to run.
Here, we prove lower- and upper-bounds for the sum time-energy and max time-energy.
We also prove the optimal max time-energy for a class of channels which includes the depolarizing channel.
A consequence of this result is that erasing information takes more time-energy resource in the quantum setting than in the classical setting.

A related concept about erasure and energy is the Landauer's principle~\cite{Landauer1961} which puts lower limits on the energy dissipated to the environment in erasing (qu)bits.
There is a difference between the erasure considered here and the erasure of the Landauer's principle.
First, we erase information by making the initial pure state more mixed, whereas
the Landauer's principle concerns resetting a possibly mixed state to a standard pure state.
Second, the Landauer's principle concerns erasure in the thermodynamic setting where temperature plays a key role.
Third, tradeoff between time and energy is implicated in our approach.

For future investigation, it is instructive to obtain the time-energy for various quantum processes such as some standard gates or algorithms, to consider this time-energy measure in the thermodynamic setting, and to explore deeper operational meaning about this measure.

\section*{Acknowledgments}%
We thank H.-K. Lo and X. Ma for enlightening discussion.
This work is supported in part by
RGC under Grant 700712P of the HKSAR Government.

\appendix

\section{Derivation of the matrix $\tilde{U}$}
\label{app-elements-of-U}

Here, we derive Eq.~\eqref{eqn-U-with-actual-elements}.
We are given that the eigenvalues of 
$\tilde{U}$ are $\exp(i \theta_1)$ and $\exp(i \theta_2)$.
The constraints are that
(i) the chord connecting $\exp(i \theta_1)$ and $\exp(i \theta_2)$ intersects $\braket{a}{b}$, and (ii) $\tilde{U}\ket{a}=\ket{b}$.

$\tilde{U}$ has the following decomposition:
\begin{align}
\label{eqn-U-elements-decomposition}
\tilde{U}
=
\exp(i \theta_1) \ket{\tilde{u}_1}\bra{\tilde{u}_1} +
\exp(i \theta_2) \ket{\tilde{u}_2}\bra{\tilde{u}_2}
\end{align}
where $\ket{\tilde{u}_j}$ is the eigenvector corresponding to $\exp(i \theta_j)$.
We assume that $\ket{\tilde{u}_j}, j=1,2$ take the following forms:
\begin{align}
\tilde{u}_1=
\begin{bmatrix}
\sqrt{z}\\
e^{i x} \sqrt{1-z}
\end{bmatrix}, \:\:
\tilde{u}_2=
\begin{bmatrix}
-e^{-i x} \sqrt{1-z}\\
\sqrt{z}
\end{bmatrix},
\end{align}
which are expressed in the basis $\{\ket{a},\ket{a^\perp}\}$.
Note that $\braket{\tilde{u}_1}{\tilde{u}_2}=0$.
Constraint (ii) implies that
\begin{align*}
\ket{b}=
\tilde{U}\ket{a} = 
&
\Big[
z e^{i \theta_1}+ (1-z) e^{i \theta_2}
\Big]
\ket{a}
+
\\
&
\left[
\sqrt{z(1-z)} e^{i x} (e^{i \theta_1}-e^{i \theta_2})
\right]
\ket{a^\perp} .
\end{align*}
Comparing this with 
\begin{align*}
\ket{b}=
\braket{a}{b} \ket{a}
+
\sqrt{1-\lvert \braket{a}{b} \rvert^2} \ket{a^\perp}
\end{align*}
obtained from Eq.~\eqref{eqn-a-perp},
we require that
\begin{align}
\label{eqn-U-elements-11entry}
r e^{i \gamma}
&=
z e^{i \theta_1}+ (1-z) e^{i \theta_2} 
\\
\label{eqn-U-elements-21entry}
\sqrt{1-r^2}
&=
\sqrt{z(1-z)} e^{i x} (e^{i \theta_1}-e^{i \theta_2}) 
\end{align}
where $r e^{i \gamma}=\braket{a}{b}$ expressed in the polar form.
We need to verify that both equations hold for some $0\leq z \leq 1$ and $x \in \IR$.
Constraint (i) means that there exists some $z$ that Eq.~\eqref{eqn-U-elements-11entry} holds.
We take this $z$ as fixed and find $x$ so that Eq.~\eqref{eqn-U-elements-21entry} holds, which
can only occur if 
$e^{i x} (e^{i \theta_1}-e^{i \theta_2}) \in \IR_+$.
We make an ansatz for $x$ by setting
\begin{align}
e^{i x}=i (-1)^s e^{-i\frac{\theta_1+\theta_2}{2}}
\end{align}
giving
\begin{align}
e^{i x} (e^{i \theta_1}-e^{i \theta_2}) 
=
2 (-1)^{s+1} \sin \frac{\theta_1-\theta_2}{2}
\label{eqn-U-elements-check0}
\end{align}
where $s=0,1$ is chosen so that this term is non-negative.
We square both sides of Eq.~\eqref{eqn-U-elements-21entry} and compare both sides.
For the LHS, $r^2$ can be obtained from Eq.~\eqref{eqn-U-elements-11entry} as follows:
\begin{align}
r^2&=|z e^{i \theta_1}+ (1-z) e^{i \theta_2}|^2
\nonumber
\\
&=
z^2+(1-z)^2+2z(1-z)\cos(\theta_1-\theta_2) .
\label{eqn-U-elements-check1}
\end{align}
Squaring the RHS of Eq.~\eqref{eqn-U-elements-21entry} gives
\begin{eqnarray*}
&&4z(1-z)\sin^2 \frac{\theta_1-\theta_2}{2}
\\
&=&
2z(1-z)[1-\cos(\theta_1-\theta_2)]
\end{eqnarray*}
which can be checked to be equal to $1-r^2$ using Eq.~\eqref{eqn-U-elements-check1}.
Therefore, Eqs.~\eqref{eqn-U-elements-11entry} and \eqref{eqn-U-elements-21entry} hold.
Finally, Eqs.~\eqref{eqn-U-elements-decomposition}--\eqref{eqn-U-elements-check0} together give Eq.~\eqref{eqn-U-with-actual-elements}.

\bibliographystyle{apsrev4-1}

\bibliography{paperdb}

\end{document}